\documentclass[preprint]{sigplanconf}

\usepackage{amssymb, amsthm, amsmath, mathabx, stmaryrd}
\usepackage{xargs}
\usepackage{hyperref}
\usepackage{mathpartir, fancyvrb, empheq, float}
\usepackage[usenames,dvipsnames]{xcolor}

\usepackage{todonotes}

\floatstyle{boxed}
\restylefloat{figure}

\newtheorem{theorem}{Theorem}[section]
\newtheorem{lemma}[theorem]{Lemma}
\newtheorem{proposition}[theorem]{Proposition}
\newtheorem{corollary}[theorem]{Corollary}

\newenvironment{definition}[1][Definition]{\begin{trivlist}
\item[\hskip \labelsep {\bfseries #1}]}{\end{trivlist}}
\newenvironment{example}[1][Example]{\begin{trivlist}
\item[\hskip \labelsep {\bfseries #1}]}{\end{trivlist}}
\newenvironment{remark}[1][Remark]{\begin{trivlist}
\item[\hskip \labelsep {\bfseries #1}]}{\end{trivlist}}

\newcommand{\agda}[0]{Agda}
\newcommand{\coq}[0]{Coq}
\newcommand{\minlog}[0]{Minlog}

\newcommand{\set}[0]{{\color{NavyBlue}\mathtt{Set}}}
\newcommand{\fin}[1][n]{{\color{NavyBlue}\mathtt{Fin}}~n}

\DeclareMathOperator{\whnorm}{{\color{NavyBlue}\mathtt{wh-norm}}}
\DeclareMathOperator{\whapp}{{\color{NavyBlue}\mathtt{wh-\hspace{-2pt}\app}}}
\DeclareMathOperator{\whpiun}{{\color{NavyBlue}\mathtt{wh-\hspace{-2pt}\piun}}}
\DeclareMathOperator{\whpide}{{\color{NavyBlue}\mathtt{wh-\hspace{-2pt}\pide}}}
\DeclareMathOperator{\whcomp}{{\color{NavyBlue}\mathtt{wh-\hspace{-2pt}\comp}}}
\DeclareMathOperator{\whmap}{{\color{NavyBlue}\mathtt{wh-map}}}
\DeclareMathOperator{\whfold}{{\color{NavyBlue}\mathtt{wh-fold}}}
\DeclareMathOperator{\whappend}{{\color{NavyBlue}\mathtt{wh-\hspace{-2pt}\append}}}

\DeclareMathOperator{\etanorm}{{\color{NavyBlue}\eta\mathtt{norm}}}
\DeclareMathOperator{\etalist}{{\color{NavyBlue}\eta\mathtt{list}}}
\DeclareMathOperator{\etaneut}{{\color{NavyBlue}\eta\mathtt{neut}}}

\DeclareMathOperator{\nfnorm}{{\color{NavyBlue}\mathtt{nf-norm}}}
\DeclareMathOperator{\nflist}{{\color{NavyBlue}\mathtt{nf-list}}}
\DeclareMathOperator{\nfneut}{{\color{NavyBlue}\mathtt{nf-neut}}}
\DeclareMathOperator{\nffold}{{\color{NavyBlue}\mathtt{nf-fold}}}
\DeclareMathOperator{\nfmap}{{\color{NavyBlue}\mathtt{nf-map}}}
\DeclareMathOperator{\nfappend}{{\color{NavyBlue}\mathtt{nf-\hspace{-2pt}\append}}}

\DeclareMathOperator{\standard}{{\color{NavyBlue}\mathtt{standard}}}

\newcommandx*\lcni[2][1=\Delta, 2=\Gamma]{#1 ~{\color{NavyBlue}\supseteq}~ #2}
\newcommandx*\incl[2][1=\Gamma, 2=\Delta]{#1 ~{\color{NavyBlue}\subseteq}~ #2}

\DeclareMathOperator{\app}{\$\hspace{-2pt}\$}
\DeclareMathOperator{\append}{+\hspace{-5pt}+}
\DeclareMathOperator{\comp}{\circ\hspace{-2pt}\circ}
\DeclareMathOperator{\piun}{\pi_1}
\DeclareMathOperator{\pide}{\pi_2}

\DeclareMathOperator{\teunit}{{\color{ForestGreen}\mathtt{tt}}}

\DeclareMathOperator{\gramdecl}{:\hspace{-2pt}:\hspace{-2pt}=~}

\newcommand{\context}[1][\type]{{\color{NavyBlue}\mathtt{Con}}(#1)}
\newcommand{\type}[1][n]{{\color{NavyBlue}\mathtt{type}}_{#1}}

\newcommand{\tybase}[1][k]{{\color{ForestGreen}\mathtt{`\alpha}_{#1}}}
\newcommand{\tyunit}[0]{{\color{ForestGreen}\mathtt{`1}}}
\newcommandx*\typrod[2][1=\sigma, 2=\tau]{#1 ~{\color{ForestGreen}\mathtt{`}\hspace{-2pt}\times}~ #2}
\newcommandx*\tyarrow[2][1=\sigma, 2=\tau]{#1 ~{\color{ForestGreen}\mathtt{`}\hspace{-3pt}\rightarrow}~ #2}
\newcommand{\tylist}[1][\sigma]{{\color{ForestGreen}\mathtt{`list}}~ #1}

\newcommandx*\typingrule[3][1=t, 2=\sigma, 3=\Gamma]{#3 ~{\color{NavyBlue}\vdash}~ #1 ~{\color{NavyBlue}\colon} #2}

\newcommand{\pointwiselift}{^s}
\newcommandx*\term[2][1=\sigma, 2=\Gamma]{#2 ~{\color{NavyBlue}\vdash}~ #1}
\newcommandx*\subst[2][1=\Gamma, 2=\Delta]{#2 ~{\color{NavyBlue}\vdash\hspace{-6pt}\pointwiselift}~ #1}
\newcommand{\idsubst}[1][\Gamma]{{\color{NavyBlue}\mathtt{id}}_{#1}}

\DeclareMathOperator{\base}{{\color{ForestGreen}\mathtt{base}}}
\DeclareMathOperator{\pop}{{\color{ForestGreen}\mathtt{pop!}}}
\DeclareMathOperator{\step}{{\color{ForestGreen}\mathtt{step}}}
\DeclareMathOperator{\ConEmpty}{{\color{ForestGreen}\varepsilon}}
\newcommandx*\ConExtend[2][1=\Gamma, 2=\sigma]{#1~{\color{ForestGreen}\cdot}~#2}

\newcommandx*\whneutral[2][1=\sigma, 2=\Gamma]{#2 ~{\color{NavyBlue}\vdash_{\mathit{whne}}}~ #1}
\newcommandx*\whnormal[2][1=\sigma, 2=\Gamma]{#2 ~{\color{NavyBlue}\vdash_{\mathit{whnf}}}~ #1}
\newcommandx*\neutral[2][1=\sigma, 2=\Gamma]{#2 ~{\color{NavyBlue}\vdash_{\mathit{ne}}}~ #1}
\newcommandx*\normal[2][1=\sigma, 2=\Gamma]{#2 ~{\color{NavyBlue}\vdash_{\mathit{nf}}}~ #1}

\newcommand{\eqdef}{\overset{\text{def}}{=}}
\newcommand{\weaken}[1][inc]{{\color{NavyBlue}\mathtt{wk}}_{\mathit{#1}}}
\newcommand{\weakene}[1][inc]{{\color{NavyBlue}\mathtt{wk}\pointwiselift}_{\mathit{#1}}}

\DeclareMathOperator{\tepair}{~{\color{ForestGreen}\mathtt{`,}}~}
\DeclareMathOperator{\telam}{{\color{ForestGreen}\mathtt{`}\lambda}}
\newcommand{\telamclosure}[1][\rho]{{\color{ForestGreen}\mathtt{`}\lambda[}#1{\color{ForestGreen}]}}
\DeclareMathOperator{\teapp}{{\color{ForestGreen}\mathtt{`\$}}}
\DeclareMathOperator{\tett}{{\color{ForestGreen}\mathtt{`\langle\rangle}}}
\DeclareMathOperator{\tepiun}{{\color{ForestGreen}\mathtt{`\hspace{-2pt}\piun}}}
\DeclareMathOperator{\tepide}{{\color{ForestGreen}\mathtt{`\hspace{-2pt}\pide}}}
\DeclareMathOperator{\temap}{{\color{ForestGreen}\mathtt{`map}}}
\DeclareMathOperator{\teappend}{~{\color{ForestGreen}\mathtt{`\hspace{-3pt}+\hspace{-5pt}+}}~}
\DeclareMathOperator{\tefold}{{\color{ForestGreen}\mathtt{`fold}}}
\DeclareMathOperator{\tenil}{{\color{ForestGreen}\mathtt{`[]}}}
\DeclareMathOperator{\tecons}{~{\color{ForestGreen}\mathtt{`:\hspace{-2pt}:}}~}

\newcommandx*\temappend[3][1=v_1, 2=n, 3=v_2]
{\temap{\hspace{-2pt}\color{ForestGreen}\langle}#1{~\color{ForestGreen},~}#2{\color{ForestGreen}\rangle\hspace{-5pt}\teappend}#3}

\DeclareMathOperator{\vappend}{{\color{NavyBlue}\NBE\hspace{-3pt}\mathtt{+\hspace{-5pt}+}}}
\DeclareMathOperator{\vmap}{{\color{NavyBlue}\NBE\hspace{-2pt}\mathtt{map}}}
\DeclareMathOperator{\vcomp}{{\color{NavyBlue}\NBE\hspace{-2pt}\circ}}
\DeclareMathOperator{\vfold}{{\color{NavyBlue}\NBE\hspace{-3pt}\mathtt{fold}}}

\newcommand{\exampleDomain}[0]{\tylist[(\protect{\typrod[\tyunit][\tybase]})]}
\newcommand{\exampleType}[0]{
  \term[\protect{\tyarrow[\exampleDomain][\exampleDomain]}][\ConEmpty]}

\DeclareMathOperator{\rules}{\beta\delta\iota\eta\nu}

\DeclareMathOperator{\reduce}{{\color{NavyBlue}\leadsto_{\rules}}}
\DeclareMathOperator{\reduces}{{\color{NavyBlue}\leadsto_{\rules}^*}}

\DeclareMathOperator{\conv}{{\color{NavyBlue}\equiv_{\rules{}}}}

\newcommand{\NBE}[0]{{\color{NavyBlue}\mathcal{M}}}
\newcommand{\NBEL}[0]{{\color{NavyBlue}\mathcal{L}}}
\newcommand{\NBEE}[0]{\NBE\pointwiselift}

\newcommand{\WHNF}[0]{{\color{NavyBlue}\NBE_{wh}}}
\newcommand{\AWHNF}[0]{{\color{NavyBlue}\NBE^{\star}_{wh}}}
\newcommand{\AWHNFL}[0]{{\color{NavyBlue}\NBEL^{\star}_{wh}}}
\newcommand{\WHNFE}[0]{{\color{NavyBlue}\NBEE_{wh}}}

\newcommandx*\whmodel[2][1=\sigma, 2=\Gamma]{{\color{NavyBlue}\WHNF}(#2 , #1)}
\newcommandx*\gluer[3][1=\sigma, 2=\whmodel, 3=\Gamma]{{\color{NavyBlue}\langle} #3 ~{\color{NavyBlue},}~ #1 {\color{NavyBlue}\rangle\hspace{-4pt}\lefttorightarrow} #2}
\newcommandx*\awhmodel[2][1=\sigma, 2=\Gamma]{{\color{NavyBlue}\AWHNF}(#2 , #1)}
\newcommandx*\awhmodellist[3][1=\sigma, 2=\Vdash_{\sigma}, 3=\Gamma]{{\color{NavyBlue}\AWHNFL}(#3 , #1, #2)}

\newcommandx*\model[2][1=\sigma, 2=\Gamma]{{\color{NavyBlue}\NBE}(#2 , #1)}
\newcommandx*\modellist[3][1=\sigma, 2=\modellistfunabs, 3=\Gamma]{{\color{NavyBlue}\NBEL}(#3 , #1, #2)}
\newcommand{\modellistfunabs}[1][\sigma]{\mathtt{M}_{#1}}
\newcommand{\modellistfun}[1][\sigma]{\model[#1][.~]}

\newcommandx*\whmodele[2][1=\Gamma, 2=\Delta]{{\color{NavyBlue}\WHNFE}(#2 , #1)}
\newcommandx*\modele[2][1=\Gamma, 2=\Delta]{{\color{NavyBlue}\NBEE}(#2 , #1)}

\newcommand{\idmodele}[1][\Gamma]{{\color{NavyBlue}\mathtt{id}}_{\NBEE ~ #1}}

\newcommand{\reify}[1][\sigma]{{\color{NavyBlue}\downarrow\hspace{-1pt}}_{#1}}
\newcommand{\listreify}[1][\sigma]{{\color{NavyBlue}\downarrow\hspace{-3pt}\downarrow\hspace{-1pt}}_{#1}}
\newcommand{\reflect}[1][\sigma]{{\color{NavyBlue}\uparrow\hspace{-1pt}}_{#1}}
\newcommand{\listreflect}[1][\sigma]{{\color{NavyBlue}\uparrow\hspace{-3pt}\uparrow\hspace{-1pt}}_{#1}}

\newcommand{\reflecte}[1][\texttt{id}]{{\color{NavyBlue}\reflect[#1]\pointwiselift}}

\DeclareMathOperator{\norm}{{\color{NavyBlue}norm}}
\DeclareMathOperator{\termeval}{{\color{NavyBlue}eval}}

\newcommandx*\awhrelation[4][1=\sigma, 2=t, 3=T, 4=\Gamma]{\awhmodel[#1][#4]
     ~{\color{NavyBlue}\ni}~ #2 ~{\color{NavyBlue}\lightning^{\star}_{wh}}~ #3}
\newcommandx*\whrelation[3][1=\sigma, 2=T, 3=\Gamma]{{\color{NavyBlue}\lightning_{wh}}(#1)(#2)}

\newcommandx*\sound[4][1=\sigma, 2=t, 3=T, 4=\Gamma]{\model[#1][#4]
     ~{\color{NavyBlue}\ni}~ #2 ~{\color{NavyBlue}\lightning}~ #3}
\newcommandx*\soundlist[6][1=\mathit{xs}, 2=\mathit{XS}, 3=\Gamma, 4=\sigma, 5=\modellistfunabs, 6=\protect{\soundlistfunabs[\_][\_]}]
{\modellist[#4][#5,#6][#3] ~{\color{NavyBlue}\ni}~ #1 ~{\color{NavyBlue}\lightning}~ #2}
\newcommandx*\soundlistfunabs[3][1=t, 2=T, 3=\sigma]{\modellistfunabs[#3]~#1~\lightning~#2}

\newcommandx*\sounde[4][1=\rho, 2=R, 3=\Gamma, 4=\Delta]{\modele[#3][#4]
     ~{\color{NavyBlue}\ni}~ #1 ~{\color{NavyBlue}\lightning\hspace{-3pt}\lightning}~ #2}

\newcommandx*\exteq[3][1=\sigma, 2=T, 3=U]{#2 ~{\color{NavyBlue}\equiv}_{#1}~ #3}
\newcommandx*\exteqlist[3][1=\sigma, 2=T, 3=U]{#2 ~{\color{NavyBlue}\equiv}_{#1}^{{\color{NavyBlue}\mathtt{`list}}}~ #3}
\newcommandx*\uniform[2][1=\sigma, 2=T]{{\color{NavyBlue}\mathtt{Uni}}_{#1}~ #2}

\begin{document}

\conferenceinfo{ICFP '13}{September 25--27, 2013, Boston}
\copyrightyear{2013}
\copyrightdata{[to be supplied]}

\titlebanner{DRAFT}        % These are ignored unless
\preprintfooter{New Equations for Neutral Terms}   % 'preprint' option specified.

\title{New Equations for Neutral Terms}
\subtitle{A Sound and Complete Decision Procedure, Formalized}

\authorinfo{Guillaume Allais \and Conor McBride}
           {University of Strathclyde}
           {\{guillaume.allais, conor.mcbride\}@strath.ac.uk}
\authorinfo{Pierre Boutillier}
           {PPS - Paris Diderot}
           {pierre.boutillier@pps.univ-paris-diderot.fr}

\maketitle

\begin{abstract}
The definitional equality of an intensional type theory is its
test of type compatibility. Today's systems rely on ordinary
evaluation semantics to compare expressions
in types, frustrating users with type errors arising
when evaluation fails to identify two `obviously' equal
terms. If only the machine could decide a richer
theory! We propose a way
to decide theories which supplement evaluation with `$\nu$-rules',
rearranging the neutral parts of normal forms, and report
a successful initial experiment.

We study a simple $\lambda$-calculus with primitive fold, map and
append operations on lists and develop in \agda\ a sound and complete
decision procedure for an equational theory enriched with monoid,
functor and fusion laws.
\end{abstract}

%~ \category{CR-number}{subcategory}{third-level}
%~ \terms
%~ term1, term2


\section{Introduction}

\newcommand{\fbc}[1]{\framebox{\hspace*{-0.04in}\texttt{\raisebox{0in}[0.10in][0.04in]{#1}}\hspace*{-0.04in}}}
\newcommand{\fbb}[1]{\framebox{\hspace*{-0.04in}\texttt{\raisebox{0in}[0.08in][0.02in]{#1}}\hspace*{-0.04in}}}
\newcommand{\fba}[1]{\framebox{\hspace*{-0.04in}\texttt{\raisebox{0in}[0.06in][0in]{#1}}\hspace*{-0.04in}}}
\DefineVerbatimEnvironment%
  {Code}{Verbatim}
  {
    commandchars=\\\{\}, codes={\catcode`$=3}
  }

The programmer working in intensional type theory is no stranger to
`obviously true' equations she wishes held \emph{definitionally} for
her program to typecheck without having to chase down ill-typed
terms and brutally coerce them. In this article, we present one way
to relax definitional equality, thus accommodating some of her
longings. We distinguish three types of fundamental relations between terms.

The first denotes computational rules: it is untyped, \emph{oriented} and
denoted by $\leadsto$ in its one step version or $\leadsto^\star$ when the
reflexive transitive congruence closure is considered. In Table
\ref{table:deltaiota}, we introduce a few such rules which correspond to the
equations the programmer writes to define functions. They are referred to as
$\delta$ (for \emph{definitions}) and $\iota$ (for pattern-matching on
\emph{inductive} data) rules and hold computationally just like the more common
$\beta$-rule.

\begin{table}[h]
\begin{Code}
map : (a $\rightarrow$ b) $\rightarrow$ list a $\rightarrow$ list b
map f []        $\leadsto$ []
map f (x :: xs) $\leadsto$ f x :: map f xs

(++) : list a $\rightarrow$ list a $\rightarrow$ list a
[]      ++ ys $\leadsto$ ys
x :: xs ++ ys $\leadsto$ x :: (xs ++ ys)

fold : (a $\rightarrow$ b $\rightarrow$ b) $\rightarrow$ b $\rightarrow$ list a $\rightarrow$ b
fold c n []        $\leadsto$ n
fold c n (x :: xs) $\leadsto$ c x (fold c n xs)
\end{Code}
\caption{\label{table:deltaiota}$\delta\iota$-rules - computational}
\end{table}

The second is the judgmental equality ($\equiv$): it is typed, tractable
for a machine to decide and typically includes $\eta$-rules for negative types
therefore internalizing some kind of \emph{extensionality}.
Table \ref{table:eta} presents such rules, explaining that some types have
unique constructors which the programmer can demand. They are well supported
in e.g. Epigram~\cite{DBLP:conf/sfp/ChapmanAM05} and \agda~\cite{agda} both
for functions and records but still lacking for records in \coq~\cite{coq}.

\begin{table}[h]
\begin{Code}
$\Gamma \vdash$ f $\equiv$ $\lambda$ x. f x      : a $\to$ b
$\Gamma \vdash$ p $\equiv$ ($\piun$ p , $\pide$ p) : a * b
$\Gamma \vdash$ u $\equiv$ ()            : 1
\end{Code}
\caption{\label{table:eta}$\eta$-rules - canonicity}
\end{table}

The third is the propositional equality ($=$): this lets us state and give
evidence for equations on open terms which may not be identified judgmentally.
Table \ref{table:nu} shows a kit for building computationally inert
\emph{neutral} terms growing layers of thwarted progress around a variable which
we dub the `nut', together with some equations on neutral terms which held only
propositionally -- until now. This paper shows how to extend the judgmental
equality with these new `$\nu$-rules'. We gain, for example, that
\texttt{map swap . map swap $\equiv$ id}, where \texttt{swap} swaps the
elements of a pair.

\[
\fba{x}\;\;\fbb{\fba{~}~a}\;\;\fbb{$\piun$~\fba{~}}\;\;\fbb{$\pide$~\fba{~}}
\;\;\fbb{\fba{~}~++~ys}\;\;\fbb{map~f~\fba{~}}\;\;\fbb{fold~n~c~\fba{~}}
\]

\begin{table}[h]
\begin{tabular}{@{}rcl@{}}
\fbb{\fba{xs}~++~[]} & = & \fba{xs} \smallskip\\
\fbc{\fbb{(\fba{xs}~++~ys)}~++~zs} & = & \fbb{\fba{xs}~++~(ys~++~zs)}
\medskip\\
\fbb{map~id~\fba{xs}} & = & \fba{xs} \smallskip\\
\fbc{map~f~\fbb{(map~g~\fba{xs})}} & = & \fbb{map~(f~.~g)~\fba{xs}}
\medskip \\
\fbc{map~f~\fbb{(\fba{xs}~++~ys)}} & = & \fbc{\fbb{map~f~\fba{xs}}~++~map~f~ys}
\medskip \\
\fbc{fold~c~n~\fbb{(map~f~\fba{xs})}} & = & \fbb{fold~(c~.~f)~n~\fba{xs}}\smallskip\\
\fbc{fold~c~n~\fbb{(\fba{xs}~++~ys)}} & = & \fbb{fold~c~(fold~c~n~ys)~\fba{xs}}
\end{tabular}
\caption{\label{table:nu}$\nu$-rules}
\end{table}

A $\nu$-rule is an equation between neutral terms with the same nut which holds
just by structural induction on the nut, with $\beta\delta\iota$ reducing
subgoals to inductive hypotheses -- the classic proof pattern of Boyer and
Moore~\cite{BoyerMoore}. Consequently, we need only use $\nu$-rules
to standardize neutral terms after ordinary evaluation stops. This separability
makes implementation easy, but the proof of its completeness correspondingly
difficult. Here, we report a successful experiment in formalizing a modified
normalization by evaluation proof for simply-typed $\lambda$-calculus with
list primitives and the $\nu$-rules above.

\paragraph{Contents} We define the terms of the theory and deliver a sound and
complete normalization algorithm in Sections \ref{experimental-setting} to
\ref{correctness-proof}. We then explain how this promising experiment can be
scaled up to type theory (Section \ref{scale-up-tt}) thus suggesting that
other frustrating equations of a similar character may soon come within our
grasp (Section \ref{further-opportunities}).

%%%%%%%%%%%%%%%%%%%%%%%%%%%%%%%%%%%%%%%%%%%%%%%%%%%%%%%%%%%%%
%%%%%%%%%%%%%%%%%%%%%%%%% Setting %%%%%%%%%%%%%%%%%%%%%%%%%%%
%%%%%%%%%%%%%%%%%%%%%%%%%%%%%%%%%%%%%%%%%%%%%%%%%%%%%%%%%%%%%

\section{Our Experimental Setting}
\label{experimental-setting}

In a dependently-typed setting, one has to deal with issues unrelated to the
matter at hand: Danielsson's formalization of a Type Theory as an
inductive-recursive family uses a non strictly positive
datatype~\cite{NadNbeDep}, Abel et al.~\cite{NbeDep1} resort to
recursive domain equations together with logical relations proving them
meaningful, McBride's proposition~\cite{McBride2010Outrageous} is only able
to steal the judgmental equality of the implementation language and Chapman's
big step formulation is not proven terminating~\cite{ChapmanPhd}.

We propose a preliminary experiment on a calculus for which the formalization
in \agda{} is tractable: we are interested in the modifications to be made to
an existing implementation in order to get a complete procedure for the extended
equational theory. We developed the algorithm during Boutillier's internship
at Strathclyde~\cite{LambList}; Allais completed the formalized meta-theory.

\paragraph{Types} The set of types is parametrized by a finite set of
base types $\tybase[1], \dots, \tybase[n]$ it can build upon. These
unanalysed base types give us a simple way to model expressions
exhibiting some parametric polymorphism.
$$\sigma, \tau, \dots \gramdecl
           \tybase
  \mid \tyunit
  \mid \typrod
  \mid \tyarrow
  \mid \tylist$$
\begin{remark} In the \agda{} implementation this
indexing by a finite set of base types is modelled by defining a
nat-indexed family $\type$ with a constructor $\tybase[]$ taking a
natural number $k$ bounded by $n$ (an element of $\fin$) to refer to
the $k^{th}$ base type.
\end{remark}

\paragraph{Terms} Terms follow the grammar presented below and the typing
rules described in Figure~\ref{typingrules} where contexts are just snoc
lists of variable names together with their type.
\begin{align*}
 t, u, \dots
    & \gramdecl x \mid \telam x. t \mid t \teapp u \mid \tett \mid t \tepair u \mid \tepiun t \mid \tepide t \mid \tenil \\
    & \mid hd \tecons tl \mid \temap(f , xs) \mid xs \teappend ys \mid \tefold (c , n , xs)
\end{align*}

\begin{figure*}
\begin{tabular}{l|r}
\begin{minipage}{0.25\textwidth}
\begin{mathpar}
\inferrule{ }{\base \colon \incl[\ConEmpty][\ConEmpty]}
\and \inferrule{\mathit{pr} \colon \incl}
     {\pop  \mathit{pr} \colon \incl[\protect{\ConExtend[\Gamma][(x \colon \sigma)]}][\protect{\ConExtend[\Delta][(x \colon \sigma)]}]}
\and \inferrule{\mathit{pr} \colon \incl}
     {\step  \mathit{pr} \colon \incl[\Gamma][\protect{\ConExtend[\Delta][(x \colon \sigma)]}]}
\end{mathpar}
\end{minipage}

& \begin{minipage}{0.75\textwidth}\begin{mathpar}
\inferrule{(x \colon \sigma) \in \Gamma}{\typingrule[x]}
\and \inferrule{\typingrule[t][\tau][\protect{\ConExtend[\Gamma][(x \colon \sigma)]}  ]}{\and \typingrule[\telam x. t][\tyarrow]}
\and \inferrule{\typingrule[t][\tyarrow] \\ \typingrule[u]}{\typingrule[t \teapp u][\tau]}
\end{mathpar}\begin{mathpar}\inferrule{ }{\typingrule[\tett][\tyunit]}
\and \inferrule{\typingrule[t] \\ \typingrule[u][\tau]}{\typingrule[t \tepair u][\typrod]}
\and \inferrule{\typingrule[t][\typrod]}{\typingrule[\tepiun t]}
\and \inferrule{\typingrule[t][\typrod]}{\typingrule[\tepide t][\tau]}
\and \inferrule{ }{\typingrule[\tenil][\tylist]}
\and \inferrule{\typingrule[hd] \\ \typingrule[tl][\tylist]}
  {\typingrule[hd \tecons tl][\tylist]}
\and \inferrule{\typingrule[xs][\tylist] \\ \typingrule[ys][\tylist]}
  {\typingrule[xs \teappend ys][\tylist]}
\and \inferrule{\typingrule[f][\tyarrow] \\ \typingrule[xs][\tylist]}
  {\typingrule[\temap(f , xs)][\tylist[\tau]]}
\and \inferrule{\typingrule[c][\tyarrow[\tyarrow]] \\ \typingrule[n][\tau]
\\ \typingrule[xs][\tylist]}{\typingrule[\tefold(c , n , xs)][\tau]}
\end{mathpar}
\end{minipage}
\end{tabular}
\caption{Context inclusion and typing rules}
\label{typingrules}
\end{figure*}
For sake of clarity in the formalization, we quote the constructors of
our object language, making a clear distinction from the corresponding
features of the host language, \agda, where we use the standard `typed
de Bruijn index' representation of well-typed
terms~\cite{deBruijn:dummies,DBLP:conf/csl/AltenkirchR99} to eliminate
junk from consideration.
In our treatment here, we always assume freshness of the variables
introduced by $\lambda$-abstractions. And we do not artificially separate
well-typed terms and typing derivations; in other words we will use
alternatively $\typingrule$ and $t \colon \term$ to denote the same objects.

\paragraph{Weakening} The notion of context inclusion
gives rise to a weakening operation $\weaken[\_]$ which can be viewed as
the action on morphisms of the functor $\term[\sigma][\_]$ from the
category of contexts and their inclusions to the category of well-typed
terms and functions between them. It is defined inductively (cf.
Figure~\ref{typingrules}) rather than as a function transporting
membership predicates from one context to its extension in order to
avoid having to use an extensionality axiom to prove two context inclusion
proofs to be the same. This more intensional presentation can already be
found under the name \textit{order preserving embeddings} in Chapman's
thesis~\cite{ChapmanPhd}.

\paragraph{From types to contexts} We can lift the notion of well-typed
terms $\term$ to whole parallel substitutions. For any two contexts named
$\Gamma$ and $\Delta$, the well-typed parallel substitution from $\Gamma$
to $\Delta$ is defined by:
$$\subst = \left\lbrace
\begin{array}{l@{\text{ if }\Gamma =~}l}
\top & \varepsilon \\
\subst[\Gamma'] \times \term[\sigma][\Delta]
  & \ConExtend[\Gamma'][(x : \sigma)]
\end{array}\right.$$
We write $t [ \rho ]$ for the application of the parallel substitution $\rho \colon \subst$
to the term $t \colon \term$ yielding a term of type $\term[\sigma][\Delta]$.

\begin{remark}
All the notions described in this document can be lifted in a pointwise
fashion to either contexts when they are defined on types or parallel
substitutions when they deal with terms. We will assume these extensions
defined and casually use the same name (augmented with: $\pointwiselift$)
for the extension and the original concept.
\end{remark}

%\begin{remark}[Remarks about the formalization:]
%\par When writing typing derivations, we always assume freshness of the variables
%introduced by lambdas. This is dealt with in the formalization by using de Bruijn
%indices. Moreover it is easier in \agda{}'s setting to work directly on terms
%well-typed by construction given that most of the proofs are done by induction on
%the typing derivation. Therefore we do not artificially separate terms and typing
%derivations.
%\end{remark}

\paragraph{Judgmental Equality} The equational theory of the calculus,
denoted $\conv$, is quite naturally the congruence closure of the
$\rules$-rules described earlier where reductions under
$\lambda$-abstraction are allowed. In this paper, we also mention
the relation $\reduces$ where the rules presented earlier are all considered with
a left to right orientation (except for the identity laws for the list functor
and the list monoid) thus inducing a notion of \emph{reduction}. The soundness
theorem proves that not only is the term produced by our normalization
procedure related to the source one but it is a reduct of it.

One easy sanity check we recommend before starting to work on the meta-theory was
to give a shallow embedding of the calculus in a pre-existing sound type theory
and to show that the reduction relation is compatible with the propositional
equality in this theory. We used \agda{} extended with a postulate stating
extensional equality for non-dependent functions in our formalization. Once the
reader is convinced that no silly mistakes were made in the equational theory,
she can start the implementation.

%%%%%%%%%%%%%%%%%%%%%%%%%%%%%%%%%%%%%%%%%%%%%%%%%%%%%%%%%%%%%
%%%%%%%%%%%%%%%%%%%%%%%%% Setting %%%%%%%%%%%%%%%%%%%%%%%%%%%
%%%%%%%%%%%%%%%%%%%%%%%%%%%%%%%%%%%%%%%%%%%%%%%%%%%%%%%%%%%%%

\section{Reduction Machinery}
\label{reduction-machinery}

When looking in details at different accounts of normalization by evaluation
~\cite{BerSch91,CoqDyb97,Coquand02,NbeEffects}, the reader should be able to
detect that there are two phases in the process:
firstly the evaluation function building elements of the model from well-typed
terms performs $\beta\delta\iota$-reductions and does not reduce under
$\lambda$-abstractions effectively building closures -- using the $\lambda$-abstractions
of the host language -- when encountering one.
Secondly the quoting machinery extracting terms from the model performs $\eta$-expansions
where needed which will cause the closures to be reduced and new computations to
be started. This two-step process was already more or less present in Berger and
Schwichtenberg's original paper~\cite{BerSch91}:
\begin{quote}Obviously each term in $\beta$-normal form may be transformed into
long $\beta$-normal form by suitable $\eta$-expansions. Therefore each term $r$
may be transformed into a unique long $\beta$-normal form $r^\star$ by
$\beta$-conversion and $\eta$-expansions.
\end{quote}

Building on this ascertainment, we construct a three (rather than two) staged
process successively performing $\beta\delta\iota$, $\eta$ and finally $\nu$
reductions whilst always potentially calling back a procedure from a preceding
stage to reduce further non-normal terms appearing when e.g. going under
$\lambda$-abstractions during $\eta$-expansion, distributing a map over an
append, etc.

\subsection{The Three Stages of Standardization}
\label{bigstep}
The normalization and standardization process goes through three successive stages
whence the need to define three different subsets of terms of our calculus. They
have to be understood simply as syntactic category restricting the shape of terms
typed in the same way as the ones in the original languages except for the few
extra constructors for which we explicitly detail what they mean.

\begin{remark}It should be noted that the two last steps never reduce a term to
a constructor-headed one for datatypes (lists in our setting). In particular,
the last step only rearranges stuck terms to produce terms which are themselves
stuck. In other words: if a term (a list in our case) is convertible to a
constructor headed term (be it either nil or cons), then it is reduced to it
by the first step of the reduction.
\end{remark}

\begin{example}We will consider the normalization of
$(\telam x. x) \teapp (\telam x. x)$ of type
$\exampleType$
as a running example demonstrating the successive steps.
\end{example}

\paragraph{Untyped $\beta\iota$-reductions} The first intermediate language we
are going to encounter is composed of weak-head $\beta\delta\iota$-normal
expressions i.e. we never reduce under a lambda, this role being assigned to
the $\eta$-expansion routine. Having $\lambda$-closures as first-class values
is one of the characteristics of this approach.

\begin{figure}[h]
\begin{empheq}{align*}
m \gramdecl & x \mid m \teapp w \mid \tepiun m \mid \tepide m \mid \tefold (w_1 , w_2 , m) \\
            & \mid \temap(w , m) \mid m \teappend w \\
w \gramdecl & m \mid \telamclosure x. t \mid \tett \mid w_1 \tepair w_2 \mid \tenil \mid w_1 \tecons w_2\\
\rho \gramdecl & \varepsilon \mid \rho , x \mapsto w
\end{empheq}
\caption{Weak-head normal forms}
\label{WHForms}
\end{figure}

These values are computed using a simple off the shelf environment machine which
returns a constructor when facing one; stores the evaluation environment in a
$\lambda$-closure when evaluating a term starting with a $\telam$; and calls an
helper function (e.g. $\whapp$, $\whpiun$, $\whpide$, etc.) on the recursively
evaluated subterms when uncovering an eliminator. These helper functions
either return a neutral if the interesting subterm was stuck or perform the
elimination which may start new computations (e.g. in the application case).
We call $\whnorm$ this evaluation function.

\begin{remark}This reduction step is absolutely type-agnostic and could therefore
be performed on terms devoid of any type information as in e.g. \coq{} where
conversion is untyped. Keeping and propagating \emph{some} types (e.g. the codomain
of the function in a map) is nonetheless needed to be able to infer back the type
of the whole expression which is crucial in the following steps.
\end{remark}

\begin{example}The untyped evaluation reduces our simple example
$(\telam x. x) \teapp (\telam x. x)$ to the usual identity function:
$\telamclosure[\teunit] x. x$.
\end{example}

\paragraph{Type-directed $\eta$-expansion}Then an $\eta$-expansion step kicks in and produces $\eta$-long values in a
type-directed way. It insists that the only neutrals worthy of being considered
normal forms are the ones of the base type. It also carves out the subset of stuck
lists in a separate syntactic category $l$ thus preparing for the last step which
will leave most of the rest of the language untouched.

\begin{figure}[h]
\begin{empheq}{align*}
n \gramdecl & x \mid n \teapp v \mid \tepiun n \mid \tepide n \mid \tefold (v_1 , v_2 , l) \\
v \gramdecl & n_{\tybase} \mid l \mid \telam x. v \mid \tett \mid v_1 \tepair v_2 \mid \tenil \mid v_1 \tecons v_2 \\
l \gramdecl & n_{\tylist} \mid \temap (v , l) \mid l \teappend v
\end{empheq}
\caption{$\eta$-long values}
\label{EtaLongValues}
\end{figure}

The $\eta$-expansion of product and function type actually calls back the
subroutines for $\beta\delta\iota$-rules projecting components out of pairs
or performing function application -- here to the variable newly introduced.
This step is the only one requiring a name generator which allows us to avoid
threading such an artifact along the whole reduction machinery.
We call $\etanorm$ the main function performing this step and present it in
Figure~\ref{etanorm}. $\etalist$ and $\etaneut$ are two trivial auxiliary
functions going structurally through either lists or neutral terms and calling
$\etanorm$ whenever necessary.

\begin{figure}[h]
$$
\begin{array}{@{\etanorm (}l@{)~t~=~}l}
  \tybase  & \etaneut t \\
  \tylist  & \etalist~ \sigma~ t \\
  \tyunit  & \tett \\
  \typrod  & \etanorm ~\sigma~ (\whpiun t) \tepair \etanorm ~\tau~ (\whpide t) \\
  \tyarrow & \telam x. \etanorm \tau (t \whapp x))
\end{array}
$$
\caption{From weak-head normal forms to $\eta$-long ones}
\label{etanorm}
\end{figure}

\begin{example}The $\eta$-expansion of the evaluated form
$\telamclosure[\teunit] x. x$ of type $\exampleType$ proceeds in
multiple steps.
\begin{itemize}
 \item The arrow type forces us to introduce a $\lambda$-abstraction:

$\telam x. \etanorm ~(\exampleDomain)~ ((\telamclosure[\teunit] x. x) \whapp x)$.
 \item Now, $(\telamclosure[\teunit] x. x) \whapp x$ trivially reduces to $x$,
a neutral of list type, left unmodified by $\eta$-expansion. Hence the
$\eta$-long form: $\telam x. x$.
\end{itemize}
\end{example}

\paragraph{$\nu$-rules reorganizing neutrals}
Standard forms have a very specific shape due to the fact that we now completely
internalize the $\nu$-rules. The new constructor $\temappend[\_][\_][\_]$ --
referred to as ``mapp'' -- has the obvious semantics that it represents the
concatenation of a stuck map and a list.

\begin{figure}[h]
\begin{empheq}{align*}
n \gramdecl & x \mid n \teapp v \mid \tepiun n \mid \tepide n \mid \tefold (v_1 , v_2 , n) \\
v \gramdecl & n_{\tybase} \mid s \mid \telam x. v \mid \tett \mid v_1 \tepair v_2 \mid \tenil \mid v_1 \tecons v_2 \\
s \gramdecl & \temappend
\caption{Standard Forms}
\label{StandardForms}
\end{empheq}
\end{figure}

The standard lists $s$ are produced by flattening the stuck map / append trees
present in $l$ after the end of the previous procedure whilst the fold / map
and fold / append fusion rules are applied in order to compute folds further
and reach the point where a stuck fold is stuck on a \emph{real} neutral lists.
These reductions are computed by the mutually defined $\nfnorm$, $\nfneut$ and
$\nflist$ respectively turning $\eta$-long normals, neutrals and lists into
elements of the corresponding standard classes. $\nfnorm$ and $\nfneut$ are
mostly structural except for the few cases described in Figure \ref{standnorm}.

We define $\standard$ as being the composition of $\etanorm$ and $\nfnorm$
whilst $\norm$ is the composition of $\whnorm$ and $\standard$. As one can
see below, $\nu$-rules can restart computations in subterms by invoking
subroutines of the evaluation function $\whnorm$. Formally proving
the termination of the whole process is therefore highly non-trivial.

\begin{figure}[h]
\begin{align*}
\nfnorm (\tylist) \mathit{xs_{ne}} &= \nflist \mathit{xs} \\
\nfneut (\tefold c ~n~ \mathit{xs}) &= \nffold c ~n~ (\nflist \mathit{xs})
\end{align*}
$$\begin{array}{@{\nflist~}l@{~=~}l}
  \mathit{xs_{ne}} & \temappend[\norm (\telam x. x)][xs][\tenil] \\
  (\temap f \mathit{xs}) & \nfmap ~f~ (\nflist \mathit{xs}) \\
  (\mathit{xs} \teappend \mathit{ys}) &
    \nfappend (\nflist \mathit{xs}) (\nfnorm ~\_~ \mathit{ys})
\end{array}$$

$$\begin{array}{ll@{~=~}l}
\multicolumn{3}{l}{\nffold c~n~(\temappend[f][xs][ys]) =
  \tefold~ \mathit{cf} ~\mathit{ih}~ \mathit{xs}} \\
\text{ where} & \mathit{cf} & \standard ~\_~ (c \whcomp f) \\
              & \mathit{ih} & \standard ~\_~ (\whfold~ c ~n~ \mathit{ys}) \\ \\

\multicolumn{3}{l}{\nfmap f~(\temappend[g][xs][ys]) = \temappend[fg][xs][fys]} \\
\text{ where} & \mathit{fg}  & \standard ~\_~ (f \whcomp g)) \\
              & \mathit{fys} & \standard ~\_~ (\whmap ~ f~ \mathit{ys}) \\ \\

\multicolumn{3}{l}{\nfappend (\temappend[f][xs][ys]) ~\mathit{zs} =
  \temappend[f][xs][yzs]} \\
\text{ where} & \mathit{yzs}  & \standard ~\_~ (ys \whappend zs)
\end{array}$$

\caption{From $\eta$-long values to standard ones}
\label{standnorm}
\end{figure}

\begin{example} $\nfnorm$ does not touch the $\lambda$-abstraction but expands
the neutral $x$ of type $\exampleDomain$ to $\temappend[\texttt{id}][x][\tenil]$
where \texttt{id} is the normal form of the identity function on
$\typrod[\tyunit][\tybase]$.

We leave it to the reader to check that:
\begin{align*}
 \texttt{id} &= \etanorm ~(\typrod[\tyunit][\tybase])~ p  \\
 \texttt{id} &= \telam p. \etanorm ~(\typrod[\tyunit][\tybase])~ p  \\
             &= \telam p. (\etanorm ~\tyunit~ (\tepiun p) \tepair
                            \etanorm ~\tybase~ (\tepide p))         \\
             &= \telam p. (\tett \tepair \tepide p)
\end{align*}
Hence the final standard form of $(\telam x. x) \teapp (\telam x. x)$:
$$\telam x. \temappend[\telam p. (\tett \tepair \tepide p)][x][\tenil]$$
\end{example}

The grammar of standard terms explicitly defines a hierarchy between stuck
functions: appends are forbidden to appear inside maps and both of them have
better not be found sitting in a fold. It is but one way to guarantee the
existence of standard forms and future extensions hopefully allowing the
programmer to add the $\nu$-rules she fancies holding definitionally will have
to make sure --for completeness' sake-- that such standard forms exist.

%%%%%%%%%%%%%%%%%%%%%%%%%%%%%%%%%%%%%%%%%%%%%%%%%%%%%%%%%%%%%
%%%%%%%%%%%%%%%%%%%%%% Formalization %%%%%%%%%%%%%%%%%%%%%%%%
%%%%%%%%%%%%%%%%%%%%%%%%%%%%%%%%%%%%%%%%%%%%%%%%%%%%%%%%%%%%%

\section{Formalization of the Procedure}
\label{formalization-procedure}

What we are interested in here is to demonstrate the decidability of the equational
theory's extension rather than explaining how to prove termination of a big step
semantics in \agda{} and rely on functional induction to prove the different
properties. The reader keen on learning about the latter should refer to James
Chapman's thesis~\cite{ChapmanPhd} where he describes a principled solution to
proving termination of big step semantics for various calculi. We, on the other
hand, will focus on the former: we opted for a version of the algorithm based, in
the tradition of normalization by evaluation, on a model construction which
basically collapses the layered stages but is trivially terminating by a structural
argument.

\paragraph{Type directed partial evaluation} (or normalization by evaluation) is
a way to compute the canonical forms by using the evaluation mechanism of the host
language whilst exploiting the available type information to retrieve terms from
the semantical objects. It was introduced by Berger and Schwichtenberg~\cite{BerSch91}
in order to have an efficient normalization procedure for \minlog{}. It has since been
largely studied in different settings:

Danvy's lecture notes~\cite{TypeDirected} review its foundations and presents its
applications as a technique to get rid of static redexes when compiling a program.
It also discusses various refinements of the na\"{i}ve approach such as the introduction
of let bindings to preserve a call-by-value semantics or the addition of extra
reduction rules\footnote{E.g. $n + 0 \leadsto n$ in a calculus where $\_+\_$ is
defined by case analysis on the first argument and this expression is therefore
stuck.} to get cleaner code generated. Our $\nu$-rules are somehow reminiscent of
this approach.

T. Coquand and Dybjer~\cite{CoqDyb97} introduced a glued model recording the partial
application of combinators in order to be able to build the reification procedure
for a combinatorial logic. In this case the na\"{i}ve approach is indeed problematic
given that the \texttt{SK} structure is lost when interpreting the terms in the
na\"{i}ve model and is impossible to get back. This was of great use in the design of
a model outside the scope of this paper computing only weak-head normal forms~\cite{ForgeAcquire}.

C. Coquand~\cite{Coquand02} showed in great details how to implement and prove
sound and complete an extension of the usual algorithm to a simply-typed
lambda calculus with explicit substitutions. This development guided our
correctness proofs.

More recently Abel et al.~\cite{NbeDep1, NbeDep2} built extensions able to
deal with a variety of type theories. Last but not least Ahman and
Staton~\cite{NbeEffects, NbeEffects2} explained how to treat calculi equipped
with algebraic effects which can be seen as an extension of the calculus of
Watkins et al.~\cite{DBLP:conf/types/WatkinsCPW03} extending judgmental
equality with equations for concurrency or Filinski's computational
$\lambda$-calculus.~\cite{Filinski2001Normalization}

\begin{remark}We will call $\normal$ the typing derivations restricted to standard
values as per the previous section's definitions and $\neutral$ the corresponding
ones for standard neutrals. Standard list will be silently embedded in standard
values: the separation of $s$ and $v$ is an important vestige of the syntactic
category $l$ of stuck lists but inlining it in the grammar yields exactly the
same set of terms.
\end{remark}

\begin{remark}Following \agda{}'s color scheme, function names and type
constructors will be typeset in {\color{NavyBlue}blue}, constructors
will appear in {\color{ForestGreen}green} and variables will be left black.
\end{remark}

\begin{definition}[Model] The model is defined by induction on the type using an auxiliary
inductive definition parametric in its arguments --which guarantees that the
definition is strictly positive therefore meaningful-- to give a semantical account
of lists. One should remember that the calculus enjoys $\eta$-rules for unit,
product and arrow types; therefore the semantical counterpart of terms with such
types need not be more complex than unit, pairs and actual function spaces.
$$\begin{array}{@{\NBE(\Gamma,~}l@{~)~}cl}
  \_ & : & \type \rightarrow \set \\
  \tyunit  & = & \top\\
  \tybase  & = & \neutral[\tybase]\\
  \typrod  & = & \model[\sigma] \times \model[\tau]\\
  \tyarrow & = & \forall \Delta, \incl \rightarrow \model[\sigma][\Delta] \rightarrow \model[\tau][\Delta]\\
  \tylist  & = & \modellist[\sigma][\modellistfun]\\
\end{array}$$
Standardization may trigger new reductions and we have therefore the obligation
to somehow store the computational power of the functions part of stuck maps. This
is a bit tricky because the domain type of such functions is nowhere related to
the overall type of the expression meaning that no induction hypothesis can be used.
Luckily these new computations are only ever provoked by neutral terms: they come
from function compositions caused by map or map-fold fusions.
\begin{mathpar}
\inferrule{\Gamma \colon \context \\ \sigma \colon \type \\ \modellistfunabs \colon \context \rightarrow \set}
{\modellist \colon \set}
\and \inferrule{ }{\tenil : \modellist}
\and \inferrule{\mathit{HD} \colon \modellistfunabs(\Gamma) \\ \mathit{TL} \colon \modellist}
{\mathit{HD} \tecons \mathit{TL} \colon \modellist}
\and \inferrule{\mathit{F} \colon \forall \Delta, \incl \rightarrow \neutral[\tau][\Delta] \rightarrow \modellistfunabs(\Delta) \\
\mathit{xs} \colon \neutral[\protect{\tylist[\tau]}] \\ \mathit{YS} \colon \modellist}
{\temappend[\mathit{F}][\mathit{xs}][\mathit{YS}] \colon \modellist}
\end{mathpar}
\end{definition}
\begin{remark}One should notice the Kripke flavour of the interpretation of
function types. It is exactly what is needed to write down a weakening operation
thus giving the entire model a Kripke-like structure.
\end{remark}
\begin{definition}[Reify and reflect] Mutually defined processes allow normal forms
$\normal$ to be extracted from elements of the model $\model$ whilst neutral
forms $\neutral$ can be turned into elements of the model.
\end{definition}
\begin{proof} Both $\reify \colon \model \rightarrow \normal$ and $\reflect \colon
\neutral \rightarrow \model$ are defined by induction on their type index $\sigma$.
\paragraph{Unit, base and product types} The unit case is trivial: the
reification process returns $\tett$ while the reflection one produces the
only inhabitant of $\top$. The base type case is solved by the embedding of
neutrals into normals on one hand and by the identity function on the other
hand. The product case is simply discharged by invoking the induction
hypotheses: the reification is the pairing of the reifications of the subterms
while the reflection is the reflection of the $\eta$-expansion of the stuck
term. We can now focus on the more subtle cases.

\paragraph{Arrow type} The function case is obtained by $\eta$-expansion both
at the term level (the normal form will start with a $\telam$) and the
semantical level (the object will be a function). It is here that the fact
that the definitions are mutual is really important.
\begin{align*}
\reify[\tyarrow] F &\eqdef \telam x. \reify[\tau] F (\_, \reflect x) \\
\reflect[\tyarrow] f &\eqdef \lambda \Delta~inc~x.\reflect[\tau](\weaken(f) \teapp \reify[\sigma] x)
\end{align*}

\paragraph{Lists} The list case is dealt with by recursion on the semantical
list for the reification process and a simple injection for the reflection
case. We write $\listreify$ and $\listreflect$ for the helper functions
performing reification and reflection on lists of type $\tylist$.
$$\begin{array}{@{\listreify}l@{~\eqdef~}l}
\tenil                                  & \tenil \\
\mathit{HD} \tecons \mathit{TL}         & \reify \mathit{HD} \tecons \listreify \mathit{TL} \\
\temappend[f][\mathit{xs}][\mathit{YS}] & \temappend[\telam x.\reify f(x)][\mathit{xs}][\listreify \mathit{YS}]
\end{array}$$
This injection corresponds to applying the identity functor and monoid law. Indeed
$\lambda \Delta \_. \reflect$ denotes the identity function and has the appropriate
type $\forall \Delta, \incl \rightarrow \neutral[\sigma][\Delta] \rightarrow
\model[\sigma][\Delta]$
to fit in the semantical list mapp constructor.
$$\listreflect xs ~\eqdef~ \temappend[\lambda \Delta \_. \reflect][\mathit{xs}][\tenil]$$
\end{proof}

\begin{example} of $\eta\nu$-expansions provoked by the reflect / reify functions:
for $\mathit{xs}$ a neutral list of type $\gdef\pub{\typrod[\tyunit][\hspace{-3pt}\tybase]}\def\Tpub{\tylist[(\pub)]}\Tpub$,
we get an expanded version by drowning it in the model and reifying it back:
$$\listreify[\pub] (\listreflect[\pub] \mathit{xs}) =
\temappend[\telam p. (\tett \tepair \tepide p)][\mathit{xs}][\tenil]$$
This showcases the $\eta$-expansion of unit, products and functions as well as the
use of the identity laws mentioned during the definition of $\listreflect$.
\end{example}

Proving that every term can be normalized now amounts to proving the existence of
an evaluation function producing a term $T$ of the model $\model[\sigma][\Delta]$
given a well-typed term $t$ of the language $\term$ and a semantical environment
$\modele$. Indeed the definition of the reflection function $\reflect$ together
with the existence of environment weakenings give us the necessary machinery to
produce a diagonal semantical environment $\modele[\Gamma][\Gamma]$ which could
then be fed to such an evaluation function.

In order to keep the development tidy and have a more modular proof of correctness,
it is wise to give this evaluation function as much structure as possible. This is
done through a multitude of helper functions explaining what the semantical
counterparts of the usual combinators of the calculus (except for lambda which,
integrating a weakening to give the model its Kripke structure, is a bit special)
ought to look like.
\begin{figure*}
\small
\begin{tabular}{lr}
\begin{minipage}{0.45\textwidth}$$
\begin{array}{l@{~\vappend\mathit{ZS}=\,}l}
\tenil                          & \mathit{ZS} \\
\mathit{HD} \tecons \mathit{TL} & \mathit{HD} \tecons (\mathit{TL} \vappend \mathit{ZS}) \\
\temappend[F][\mathit{xs}][\mathit{YS}]
                                & \temappend[\mathit{F}][\mathit{xs}][(\mathit{YS} \vappend \mathit{ZS})]
\end{array}$$

$$\begin{array}{@{\vmap~\mathit{F}~}l@{~=~}l}
\tenil                          & \tenil \\
(\mathit{HD} \tecons \mathit{TL}) & F (\_, \mathit{HD}) \tecons \vmap F \, \mathit{TL} \\
(\temappend[G][\mathit{xs}][\mathit{YS}])
                                & \temappend[F \vcomp G][\mathit{xs}][\vmap F \, \mathit{YS}] \\
\multicolumn{2}{l}{\text{where}~ F ~\vcomp~ G = \lambda E~\mathit{inc}~t. F (\mathit{inc}, G(\mathit{inc},t))}
\end{array}$$

$$\begin{array}{l@{~=~}l}
\vfold~~C~N~\tenil                            & N \\
\vfold~~C~N~(\mathit{HD} \tecons \mathit{TL}) & C (\_,\mathit{HD},\_, \vfold~C~N~\mathit{TL}) \\
\vfold_{\tau} C~N~(\temappend[F][\mathit{xs}][\mathit{YS}])
                                  & \reflect[\tau] \tefold (c , n , \mathit{xs}) \\
\multicolumn{2}{l}{
\begin{array}{ll@{~=~}l}
\text{where} & c & \telam x. \telam y. \reify[\tau] C(\_, F(\_, x)), \_ , \reflect[\tau] y) \\
             & n & \reify[\tau] \vfold~C~N~\mathit{YS}
\end{array}}
\end{array}$$
\end{minipage}
& \begin{minipage}{0.55\textwidth}$$
\begin{array}{@{\termeval~}l@{~R~=~}l}
x    & R(x) \\
(\telam x. t  )& \lambda E~\mathit{inc}~S. \termeval~t ~(\weakene(R) , x \mapsto S) \\
(f \teapp x   )& (\termeval~f~R) (\_, \termeval~x~R) \\
(\tett        )& \teunit \\
(a \tepair b  )& \termeval~a~R, \termeval~b~R\\
(\tepiun t    )& \piun (\termeval~t~R) \\
(\tepide t    )& \pide (\termeval~t~R) \\
(\tenil       )& \tenil \\
(\mathit{hd} \tecons \mathit{tl})& (\termeval~\mathit{hd}~R) \tecons (\termeval~\mathit{tl}~R) \\
(\mathit{xs} \teappend \mathit{ys})& (\termeval~\mathit{xs}~R) \vappend (\termeval~\mathit{ys}~R) \\
(\temap (f , \mathit{xs}) )& \vmap (\termeval~f~R) (\termeval~\mathit{xs}~R) \\
(\tefold (c , n , \mathit{xs}) )& \vfold (\termeval~c~R) (\termeval~n~R) (\termeval~\mathit{xs}~R)
\end{array}$$
\end{minipage}
\end{tabular}
\caption{Evaluation function and semantical counterparts of list primitives}
\label{evalfun}
\end{figure*}
\begin{theorem}[Evaluation function] Given a term in $\term$ and a semantical
environment in $\modele$, one can build a semantical object in $\model[\sigma][\Delta]$.
\end{theorem}
\begin{proof}A simple induction on the term to be evaluated using the semantical
counterparts of the calculus' combinators to assemble semantical objects obtained
by induction hypotheses discharges most of the goals. See Figure~\ref{evalfun}
for the details of the code.

In the lambda case, we have the body of the lambda $b$ in $\term[\tau][\Gamma \cdot \sigma]$,
an evaluation environment $R$ in $\modele$ and we are given a context $E$, a proof
\texttt{inc} that $\incl[\Delta][E]$ and an object $S$ living in $\model[\sigma][E]$.
By combining $S$ and a weakening of $R$ along \texttt{inc}, we get an evaluation
environment of type $\modele[\Gamma \cdot \sigma][E]$ which is just what we needed
to conclude by using the $\model[\tau][E]$ delivered by the induction hypothesis
on $b$.
\end{proof}

\begin{remark}Unlike traditional normalization by evaluation, reflection and
reification are used when defining the interpretation of terms in the model.
This is made necessary by the presence of syntactical artifacts (stuck lists)
in the mapp constructor. Growing the spine of stuck eliminators calls for the
reification of these eliminators' parameters and the reflection of the whole
stuck expression to re-inject it in the model.

This kind of patterns also appeared in the glueing construction introduced by
Coquand and Dybjer in their account of normalization by evaluation for
the simply-typed SK-calculus~\cite{CoqDyb97} and can be observed in other
variants of normalization by evaluation deciding more exotic equational
theories e.g. having $\beta$-reduction but no $\eta$-rules for the simply-typed
$\lambda$-calculus~\cite{ArranSlides}.
\end{remark}

\begin{remark}The only place where type information is needed is when
reorganizing neutrals following $\nu$-rules e.g. in the semantical fold.
The evaluation function is therefore faithful to the staged evaluation approach.
The model is indeed related to the algorithm presented earlier on in
section~\ref{bigstep}: we \emph{describe} all the computations eagerly for
\agda{} to see the termination argument but a subtle evaluation strategy
applied to the produced code could reclaim the behaviour of the layered
approach. It would have to form lambda closures in the arrow case, fire
eagerly only the reductions eliminating constructors in the $\vmap$,
$\vappend$ and $\vfold$ helper functions thus postponing the execution
of the code corresponding to $\eta\nu$-rules to reification time.
\end{remark}

\begin{corollary}There is a normalization function $\norm$ turning terms in
$\term$ into normal forms in $\normal$.
\end{corollary}
\begin{proof} Given $t$ a term of type $\term$ and $\reflecte$ the function
turning a context $\Gamma$ into the corresponding diagonal semantical
environment $\modele[\Gamma][\Gamma]$, the normalization procedure is given
by the composition of evaluation and reification:
$$\norm t \eqdef \reify (\termeval (t , \reflecte \Gamma))$$
\end{proof}

%%%%%%%%%%%%%%%%%%%%%%%%%%%%%%%%%%%%%%%%%%%%%%%%%%%%%%%%%%%%%
%%%%%%%%%%%%%%%%%%%%%%% Correctness %%%%%%%%%%%%%%%%%%%%%%%%%
%%%%%%%%%%%%%%%%%%%%%%%%%%%%%%%%%%%%%%%%%%%%%%%%%%%%%%%%%%%%%

\section{Correctness}
\label{correctness-proof}

The typing information provided by the implementation language guarantees that
the procedure computes terms in normal forms from its inputs and that they have
the same type. This is undoubtedly a good thing to know but does not forbid all
the potentially harmful behaviours: the empty list is a type correct normal form
for any input of type list but it certainly is not a satisfactory answer with
respect to $\rules$-equality. Hence the need for a soundness and a completeness
theorem tightening the specification of the procedure.

The meta-theory is an ad-hoc extension of the techniques already well explained
by Catarina Coquand~\cite{Coquand02} in her presentation of a simply-typed lambda
calculus with explicit substitutions (but no data). Soundness is achieved through
a simple logical relation while completeness needs two mutually defined notions
explaining what it means for elements of $\NBE$ to be semantically equal and to
behave uniformly on extensionally equal terms.

The reader should think of these logical relations as specifying requirements for
a characterization (being equal, being uniform) to be true of an element at some
type. The natural deduction style presentation of these recursive functions should
then be quite natural for her: read in a bottom-top fashion, they express that
the (dependent) conjunction of the hypotheses -- the empty conjunction being $\top$--
is the requirement for the goal to hold. Hence leading to a natural interpretation:
\begin{mathpar}\inferrule{}{\inferrule{A \\ B \\ C}{F(t)} \\ {\huge\leadsto} \\ F(t) = A \times B \times C}\end{mathpar}

%%%%%%%%%%%%%%%%%%%%%%%%%%%%%%%%%%%%%%%%%%%%%%%%%%%%%%%%%%%%%
\subsection{Soundness}

Soundness amounts to re-building the propositional part of the reducibility
candidate argument~\cite{Girard06} which has been erased to get the bare bones
model. The logical relation $\sound$ relates a semantical object $T$ in $\model$
and a term $t$ in $\term$ which is morally the source of the semantical object.

\begin{definition}[Logical Relation for Soundness]$\sound$ is defined by
induction on the type $\sigma$ plus an appropriate inductive definition for
the list case $\soundlist$. Here are the formation rules of these types.

\begin{mathpar}
\inferrule{\mathit{t} \colon \term \\
  \mathit{T} \colon \model}
{\sound \colon \set}

\and \inferrule{\mathit{xs} \colon \term[\tylist] \\
  \mathit{XS} \colon \modellist \\
  \soundlistfunabs[\_][\_] \colon \forall \Gamma, \term \rightarrow
    \modellistfunabs \Gamma \rightarrow \set}
{\soundlist \colon \set}
\end{mathpar}

\begin{remark}
It should be no surprise to the now experienced reader that the inductive
definition of the logical relation for $\tylist$ is parametrized by
$\soundlistfunabs[\_][\_]$, the logical relation for elements of type
$\sigma$ which will be lifted to lists, simply to avoid positivity problems.
It is ultimately instantiated with the logical relation taken at type
$\sigma$.

She will also have noticed that the uses of both $\color{NavyBlue}{\mathcal{M}}$
and $\color{NavyBlue}{\mathcal{L}}$ on the left of $\color{NavyBlue}{\ni}$ are
but syntactical artifacts to hint at the connection with the model definition.
Hence the different arity in the case of the logical relation for lists.
\end{remark}

\paragraph{Unit, base and product types} The unit and base type cases are, as
expected, the simplest ones and the product case is not very much more exciting:
\begin{mathpar}
\inferrule{ }{\sound[\tyunit]}
\and \inferrule{t \reduces T}{\sound[\tybase]}
\end{mathpar}
$$\inferrule{a \colon \term \\ b \colon \term[\tau] \\ t \reduces a \tepair b
\\ \sound[\sigma][a][A] \\ \sound[\sigma][b][B]}{\sound[\typrod][t][A , B]}$$

\paragraph{Arrow type} Function types on the other hand give rise to a
Kripke-like structure in two ways: in addition to the quantification on all
possible future context which we need to match the model construction, there
is also a quantification on all possible source term reducing to the current
one.
$$\inferrule{\forall \Delta (inc \colon \incl) ~ x ~ X,
    \sound[\sigma][x][X][\Delta] \rightarrow \\
    \forall t, t \reduces \weaken f \teapp x \rightarrow
    \sound[\tau][t][F (\mathit{inc} , X)][\Delta]}
{\sound[\tyarrow][f][F]}$$

\paragraph{Lists} The cases for nil and cons are simply saying that the source
term indeed reduces to a term with the corresponding head-constructors and that
the eventual subterms are also related to the sub-objects.
\begin{mathpar}
\inferrule{t \reduces \tenil}{\soundlist[t][\tenil]}
\and \inferrule{t \reduces \mathit{hd} \tecons \mathit{tl} \\
  \soundlistfunabs[\mathit{hd}][\mathit{HD}] \\ \soundlist[\mathit{tl}][\mathit{TL}]}
{\soundlist[t][\mathit{HD} \tecons \mathit{TL}]}
\end{mathpar}
The mapp case is a bit more complex. The source term is expected to reduce to
a term with the same canonical shape and then we expect the semantical function
to behave like the one discovered.
$$\inferrule{t \reduces \temap (f , \mathit{xs}) \teappend \mathit{ys} \\
\forall \Delta (\mathit{inc} \colon \incl) ~t \rightarrow \soundlistfunabs[\weaken(f) \teapp t][F (\mathit{inc} , t)] \\
\soundlist[\mathit{ys}][\mathit{YS}]}{\soundlist[t][\protect{\temappend[F][\mathit{xs}][\mathit{YS}]}]}$$
\end{definition}

The first thing to notice is that whenever two objects are related by this
logical relation then the property of interest holds true i.e. the semantical
object indeed is a reduct of the source term. This result which mentions the
reifying function has to be proven together with the corresponding one about
the mutually defined reflection function.

\begin{definition}[Pointwise extension] We denote by $\sounde[\_][\_][\_][\_]$
the pointwise extension of the soundness logical relation to parallel
substitutions and semantical environments.
\end{definition}

\begin{lemma}Reflect and reify are compatible with this logical relation in the
sense that:
\begin{enumerate}
  \item If $t_{\mathit{ne}}$ is a neutral $\neutral$ then
     $\sound[\sigma][t_{\mathit{ne}}][\reflect t_{\mathit{ne}}]$.
  \item If $t$ and $T$ are such that $\sound$ then $t \reduces \reify T$
\end{enumerate}
\end{lemma}

The Kripke-style structure we mentioned during the definition of the logical
relation adds just what is need to have it closed under anti-reductions of the
source term:
\begin{proposition}For all $s$ and $t$ in $\term$, if $s \reduces t$ then
for all $T$ such that $\sound$, it is also true that $\sound[\sigma][s]$
\end{proposition}

The proof of soundness then mainly involves showing that the semantical counterparts
of the language's combinators we defined during the model construction are compatible
with the logical relation. Namely that e.g. if $\sound[\tyarrow][f][F]$ and
$\sound[\tylist][xs][XS]$ hold then it is also true that:
$\sound[\protect{\tylist[\tau]}][\temap(f ,xs)][\temap(F , XS)]$.

\begin{theorem}Given a term $t \colon \term$, a parallel substitution $\rho \colon \subst$
and an evaluation environment $R$ such that $\rho$ and $R$ are related ($\sounde$
holds), the evaluation of $t$ in $R$ is related to $t [\rho]$:
$\sound[\sigma][\protect{t[\rho]}][\termeval (t , R)][\Delta]$
\end{theorem}
\begin{proof}The theorem is proved by structural induction on the shape of the
typing derivation of $t$. The variable case is trivially discharged by using the
proof of $\sounde$.

All the other cases -- except for the lambda one -- can be solved by combining
induction hypotheses with the appropriate lemma proving that the corresponding
semantical combinator respects the logical relation.

In the case where $t = \telam x. b$, we are given a context $E$ together with a
proof $inc$ that it is an extension of $\Delta$, a term $u$ and an object $U$
which are related $\sound[\sigma][u][U][E]$ and, finally, a term $s \colon \term[\tau][E]$
which reduces to $(\telam x. b) [\rho] \teapp u$.
First of all, we should notice that $s \reduces b [\rho , x \mapsto u]$
and therefore that to prove $\sound[\tau][s][T][E]$ it is enough to prove that
$\sound[\tau][\protect{b [\rho , x \mapsto u]}][T][E]$. And we get just that by
using the induction hypothesis with the related parallel substitution $\rho'$ and
evaluation environment $R'$ obtained by the combination of the weakening of $\rho$
(resp. $R$) along $inc$ with $u$ (resp. $U$).
\end{proof}

\begin{corollary}A term $t$ reduces to the normal form produced by the
normalization by evaluation procedure: $t \reduces \norm t$. And if two terms
$t$ and $u$ have the same normal form up-to $\alpha$-equivalence then they are
indeed related: $t \conv u$.
\end{corollary}
\begin{proof}The identity parallel substitution is related to the diagonal evaluation
environment and $t [\idsubst]$ is equal to $t$ hence, by the previous theorem,
$\sound[\sigma][t][\termeval(t , \idmodele)]$ and then $t \reduces \norm t$.
\end{proof}

%%%%%%%%%%%%%%%%%%%%%%%%%%%%%%%%%%%%%%%%%%%%%%%%%%%%%%%%%%%%%
\subsection{Completeness}

Completeness can be summed up by the fact that the interpretation of
$\rules$-convertible elements produces semantical objects behaving similarly.
This notion of similar behaviour is formalized as \emph{semantic equality}
where, in the function case, we expect both sides to agree on any \emph{uniform}
input rather than any element of the model. As usual the list case is dealt
with by using an auxiliary definition parametric in its "interesting" arguments.

\begin{definition} The semantic equality of two elements $T, U$ of $\model$
is written $\exteq$ while $T \in \model$ being uniform is written $\uniform$.
They are both mutually defined by induction on the index $\sigma$ in
Figure~\ref{logreleq}.
\begin{figure*}
\begin{mathpar}
\inferrule{ }{\exteq[\tyunit]}
\and \inferrule{ }{\uniform[\tyunit]}

\and \inferrule{T = U}{\exteq[\tybase]}
\and \inferrule{ }{\uniform[\tybase]}

\and \inferrule{\exteq[\sigma][A][C] \\ \exteq[\tau][B][D]}{\exteq[\typrod][(A , B)][(C , D)]}
\and \inferrule{\uniform[\sigma][A] \\ \uniform[\tau][B]}{\uniform[\typrod][(A , B)]}

\and \inferrule{\forall \Delta (inc : \incl) (S : \model[\sigma][\Delta]) \rightarrow \uniform[\sigma][S] \rightarrow \exteq[\tau][F(\mathit{inc} , S)][G(\mathit{inc} , S)]}{\exteq[\tyarrow][F][G]}

\and \inferrule{\forall \Delta (inc : \incl), \uniform[\sigma][S] \rightarrow \uniform[\tau][F(inc , S)] \\
\\ \forall \Delta (\mathit{inc} : \incl) \rightarrow \uniform[\sigma][S_1] \rightarrow \uniform[\sigma][S_2] \rightarrow
\exteq[\sigma][S_1][S_2] \rightarrow \exteq[\tau][F(\mathit{inc} , S_1)][F(\mathit{inc} , S_2)] \\
\forall \mathit{inc_1}, \mathit{inc_2} \rightarrow \uniform[\sigma][S] \rightarrow
\exteq[\tau][\protect{\weaken[inc_1]~F(\mathit{inc_2} , S)}][\protect{F(\mathit{inc_2} \cdot \mathit{inc_1} , \weaken[inc_1]~ S)}]}
  {\uniform[\tyarrow][F]}

\end{mathpar}
\begin{mathpar}

\and \inferrule{ }{\tenil \colon \exteqlist[\sigma][\tenil][\tenil]}
\and \inferrule{ }{\uniform[\tylist][\tenil]}

\and \inferrule{\mathit{hd} \colon \mathit{EQ_\sigma} (X , Y) \\ \mathit{tl} \colon \exteqlist[\sigma][\mathit{XS}][\mathit{YS}]}
  {\mathit{hd} \tecons \mathit{tl} \colon \exteqlist[\sigma][X \tecons \mathit{XS}][Y \tecons \mathit{YS}]}
\and \inferrule{\uniform[\sigma][\mathit{HD}] \\ \uniform[\tylist][\mathit{TL}]}{\uniform[\tylist][\mathit{HD} \tecons \mathit{TL}]}

\and \inferrule{\mathit{xs} \colon \mathit{xs_1} \equiv \mathit{xs_2} \\
\mathit{YS} \colon \exteqlist[\sigma][\mathit{YS_1}][\mathit{YS_2}] \\
F \colon \forall \Delta (\mathit{inc} : \incl) (t : \neutral[\tau][\Delta]), \mathit{EQ_\sigma}(F_1 (\mathit{inc} , t) , F_2 (\mathit{inc} , t))}
  {\temap(F , \mathit{xs}) \teappend \mathit{YS} \colon \\ \exteqlist[\sigma][\temap(F_1 , \mathit{xs_1}) \teappend \mathit{YS_1}][\temap(F_2 , \mathit{xs_2}) \teappend \mathit{YS_2}]}

\and \inferrule{\forall \Delta (\mathit{inc} \colon \incl) (t \colon \neutral[\tau][\Delta]) , \uniform[\sigma][F(\mathit{inc} , t)] \\
\forall \mathit{inc_1}, \mathit{inc_2}, t, \exteq[\sigma][\protect{\weaken[inc_1]~F(\mathit{inc_2} , t)}][\protect{F(\mathit{inc_2} \cdot \mathit{inc_1} , \weaken[inc_1]~ t)}] \\
\uniform[\tylist][\mathit{YS}]}
{\uniform[\tylist][\temap(F , \mathit{xs}) \teappend \mathit{YS}]}
\end{mathpar}
\caption{Semantic equality and uniformity of objects in the model}
\label{logreleq}
\end{figure*}

Quite unsurprisingly, the unit case is of no interest: all the semantical units
are equivalent and uniform. Semantic equality for elements with base types is
up-to $\alpha$-equivalence: inhabitants are just bits of data (neutrals) which
can be compared in a purely syntactical fashion because we use nameless terms.
They are always uniform.

In the product case, the semantical objects are actual pairs and the definition
just forces the properties to hold for each one of the pair's components.

The function type case is a bit more hairy. While extensionality on uniform arguments is
simple to state, uniformity has to enforce a lot of invariants: application of uniform
objects should yield a uniform object, application of extensionally equal uniform objects
should yield extensionally equal objects and weakening and application should commute (up
to extensionality).

In the $\tylist$ case, extensional equality is an inductive set basically building
the (semantical) diagonal relation on lists of the same type. It is parametrized
by a relation $EQ_\sigma$ on terms of type $\model[\sigma][\Delta]$ (for any context
$\Delta$) which is, in the practical case instantiated with $\exteq[\sigma][\_][\_]$
as one would expect. Uniformity is, on the other hand, defined by recursion on the
semantical list. It could very well be defined as being parametric in something
behaving like $\uniform[\sigma][\_]$ but this is not necessary: there are no positivity
problems! It is therefore probably better to stick to a lighter presentation here.
The empty list indeed is uniform. A constructor-headed list is said to be uniform
if its head of type $\model$ is uniform and its tail also is uniform. The criterion
for a stuck list is a bit more involved. Mimicking the definition of uniformity for
functions, there are two requirements on the stuck map: applying it to a neutral
yields a uniform element of the model and application and weakening commute.
Lastly the second argument of the stuck append should be uniform too.
\end{definition}

\begin{remark}The careful reader will already have noticed that this defines a
family of equivalence relations; we will not make explicit use of reflexivity,
symmetry and transitivity in the paper but it is fundamental in the formalization.
\end{remark}

Recall that the completeness theorem was presented as expressing the fact that
elements equivalent with respect to the reduction relation were interpreted as
semantical objects behaving similarly. For this approach to make sense, knowing
that two semantical objects are extensionally equal should immediately imply
that their respective reifications are syntactically equal. Which is the case.

\begin{lemma}Reification, reflection and weakenings are compatible with the
notions of extensional equality and uniformity.
\begin{enumerate}
  \item If $\exteq$ then $\reify T = \reify U$
  \item If $t_{ne}$ is a neutral $\neutral$ then $\uniform[\sigma][(\reflect t_{ne})]$
  \item Weakening and reification commute for uniform objects
\end{enumerate}
\end{lemma}

Now that we know that all the theorem proving ahead of us will not be meaningless,
we can start actually tackling completeness. When applying an extensional function,
it is always required to prove that the argument is uniform. Being able to certify
the uniformity of the evaluation of a term is therefore of the utmost importance.

\begin{lemma} Evaluation preserves properties of the evaluation environment.
\begin{enumerate}
  \item Evaluation in uniform environments produces uniform values
  \item Evaluation in semantically equivalent environments produces semantically equivalent values
  \item Weakening the evaluation of a term is equivalent to evaluating this term in a weakened environment
\end{enumerate}
\end{lemma}

\begin{theorem}If $s$ and $t$ are two terms in $\term$ such that $s \reduce t$
and if $R$ is a uniform environment in $\modele$ then
$\exteq[\sigma][\termeval(s , R)][\termeval(t , R)]$.
\end{theorem}
\begin{proof}One proceeds by induction on the proof that $s$ reduces to $t$.
\paragraph{Structural rules} Structural rules can be discharged by combining
induction hypotheses and reflexivity proofs using previously proved lemma such
as the fact that evaluation in uniform environments yields uniform elements
for the structural rule for the argument part of application.

\paragraph{$\beta\iota$-rules} Each one the $\iota$ rules holds by reflexivity
of the extensional equality, indeed evaluation realizes these computation rules
syntactically. The case of the $\beta$ rule is slightly more complicated.
Given a function $\telam x. b$ and its argument $u$, one starts by proving
that the diagonal semantical environment extended with the evaluation of $u$
in $R$ is extensionally equal to the evaluation in $R$ of the diagonal
substitution extended with $u$. Thence, knowing that the evaluations of a term
in two extensionally equal environments are extensionally equal, one can see
that the evaluation of the redex is related to the evaluation of the body in
an environment corresponding to the evaluation of the substitution generated
when firing the redex. Finally, the fact that $\termeval$ and substitution
commute (up-to-extensionality) lets us conclude.

\paragraph{$\eta\nu$-rules} definitely are the most work-intensive ones: except
for the ones for product and unit types which can be discharged by reflexivity
of the semantic equality, all of them need at least a little bit of theorem
proving to go through. It is possible to prove the map-id, map-append,
append-nil, associativity of append and various fusion rules by induction on
the `nut' for uniform values. Solving the goals is then just a matter of
combining the right auxiliary lemma with facts proved earlier on, typically
the uniformity of semantical object obtained by evaluating a term in a uniform
environment.
\end{proof}

\begin{corollary}[Completeness] For all terms $t$ and $u$ of type $\term$, if
$t \conv u$ then $\norm t = \norm u$.
\end{corollary}
\begin{proof}Reflection produces uniform values and uniformity is preserved
through weakening hence the fact that the trivial diagonal environment is uniform.
Combined with iterations of the previous lemma along the proof that $t \conv u$,
we get that the respective evaluations of $t$ and $u$ are extensionally equal which
we have proved to be enough to get syntactically equal reifications.
\end{proof}

\begin{corollary}The equational theory enriched with $\nu$-rules is decidable.
\end{corollary}
\begin{proof}Given terms $t$ and $u$ of the same type $\term$, we can get two normal
forms $t_{nf} = \norm t$ and $u_{nf} = \norm u$ and test them for equality up-to
$\alpha$-conversion (which is a simple syntactic check in our nameless representation
in \agda{}).

If $t_{nf} = u_{nf}$ then the soundness result allows us to conclude that $t$
and $u$ are convertible terms.

If $t_{nf} \neq u_{nf}$ then $t$ and $u$ are not convertible. Indeed, if they
were then the completeness result guarantees us that $t_{nf}$ and $u_{nf}$ would
be equal which leads to a contradiction.
\end{proof}

\begin{example} of terms which are identified thanks to the internalization of
the $\nu$-rules.
\begin{enumerate}
  \item In a context with two functions $f$ and $g$ of type $\tyarrow[\sigma][\tyunit]$,
  $\telam xs. \temap (f , xs)$ and $\telam xs. \temap (g , xs)$ both normalize to
  $\telam xs. \temap (\telam \_. \tett, xs) \teappend \tenil$ and are therefore
  declared equal.
  \item At type $\def\Two{(\typrod[\tybase][\protect{\tybase[l]}])}
  \def\Twos{\tylist[\Two]}\term[\protect{\tyarrow[\Twos][\Twos]}]$,
  the terms $\telam xs. xs$ and $\telam xs. \temap (swap , \temap (swap , xs))$
  where $swap$ is the function $\telam p. (\tepide p \tepair \tepiun p)$ swapping
  the order of a pair's elements are convertible with normal form
  $\telam xs. \temap (\telam p. (\tepiun p \tepair \tepide p) , xs) \teappend \tenil$.
\end{enumerate}
\end{example}

%%%%%%%%%%%%%%%%%%%%%%%%%%%%%%%%%%%%%%%%%%%%%%%%%%%%%%%%%%%%%
%%%%%%%%%%%%%%%%%%%%%%% Future work %%%%%%%%%%%%%%%%%%%%%%%%%
%%%%%%%%%%%%%%%%%%%%%%%%%%%%%%%%%%%%%%%%%%%%%%%%%%%%%%%%%%%%%

\section{Scaling up to Type Theory}
\label{scale-up-tt}

Now that we know for sure that the judgmental equality can be safely extended
with some $\nu$-rules, we are ready to tackle more complex type theories.
We have already experimented with extending our simply-typed setting to a
universe of polynomial datatypes with map and fold.
We have to identify which parts of the setting are key to the success of this
technique and how to enforce that the generalized version still has good
properties.

\paragraph{Types} Arrow types will be replaced by $\Pi$-types and product
types by $\Sigma$-types but the basic machinery of evaluation and type-directed
$\eta$-expansion work in much the same way.

In Type Theory, it is not quite enough to be able to decide the judgmental
equality. Pollack's PhD thesis (\cite{phd:pollack94a}, Section 5.3.1), taught
us how to turn the typing relation with a conversion rule into a
syntax-directed typechecking algorithm by relying on ordinary evaluation (cf.
the application typing rule in Figure \ref{syndirectapp}). It is therefore
quite crucial for ensuring the reusability of previous typechecking algorithms
to be able to guarantee that ordinary evaluation is complete for uncovering
constructor-headed terms i.e. $\Gamma \vdash t \equiv C~ \vec{t_i} \colon T$
should imply that $t \leadsto^\star C~ \vec{t_i}'$.
This can be enforced by making sure that candidates for $\nu$-rules are only
reorganizing spines of stuck eliminators and are absolutely never emitting new
constructors.

\begin{figure}[h]
$$\inferrule{\Gamma \vdash f \colon F \and F \leadsto^\star (x : S) \rightarrow T \\
\Gamma \vdash s \colon S' \\
\Gamma \vdash S \equiv S' \colon \mathtt{Set}}
{\Gamma \vdash f s \colon T [s / x]}$$
\caption{Syntax-directed typing rule for application,
Pollack~\cite{phd:pollack94a}}
\label{syndirectapp}
\end{figure}

\paragraph{$\eta$-rules} A Type Theory does not need to have judgmental
$\eta$-rules for the $\nu$-rules to make sense. However this partially defeats
the purpose of this extension: without $\eta$-rules for products we fail to
identify the silly identity on lists of products \texttt{map swap . map swap}
with the more traditional one $\lambda x. x$ because $f_1 = \telam x. x$ is
different from $f_2 = \telam x. (\tepiun x \tepair \tepide x)$ when both terms
would reduce respectively to $\telam x. \temappend[f_1][x][\tenil]$ and
$\telam x. \temappend[f_2][x][\tenil]$. So close yet so far away!

\paragraph{Defined symbols} In this presentation, a handful of functions are
built-in rather than user-defined. This will probably be one of the biggest
changes when moving to a usable Type Theory.  We can enforce that functions
defined by pattern-matching have a fixed arity and are always fully applied
at that arity. Such a function is stuck if it is strict in a neutral argument.
Some type theories reduce pattern matching to the primitive elimination
operator for each datatype. To apply $\nu$-rules, we need to detect which
stuck eliminators correspond to which stuck pattern matches. This is the
same problem as producing readable output from normalizing open terms,
and it has already been solved by the `labelled type' translation used
in Epigram, which effectively inserts documentation of stuck pattern
matches into spines of stuck eliminators~\cite{epigram}.

\paragraph{Criteria for $\nu$-rules} Working in a setting where the datatypes
are given by a universe~\cite{LevitationPaper}, we should at least expect that
built-in generic operators, e.g. map, have associated $\nu$-rules. However, it
is clearly desirable to allow the programmer to propose $\nu$-rules for
programs of her own construction. How will the machine check that proposed
$\nu$-rules keep evaluation canonical and judgmental equality consistent and
decidable? We have already seen that $\nu$-rules must avoid to emit new
constructors; this can be summed up by the \textit{mantra}: ``A $\nu$-rule may
restart computation \emph{within} its contractum but \emph{never} in its
enclosing context''.

The candidates for $\nu$-rules should hold trivially by a Boyer-Moore style
induction; in other words, the $\beta\delta\iota-\nu$ critical pairs should be
convergent. This tells us that these rules are consistent and can be delayed
until after evaluation.

Obviously, the $\nu-\nu$ critical pairs should also be convergent. These three
criteria are all easy to check provided that $\nu$-reductions give rise to a
terminating term rewrite system.

This termination requirement is the last criterion. As a first instance, a
rather conservative approach could be to ask the user for a linear order on
defined symbols which we would lift to expressions by using the lexicographic
ordering of the encountered defined symbols starting from the ``nut'' and going
outwards. If this ordering is compatible with a left to right orientation of
the $\nu$-rules she wants to hold, then it is terminating. In the set of
$\nu$-rules used as an example in this paper, the simple ordering
$\teappend > \temap > \tefold$ is compatible with the rules.

\section{Further Opportunities for $\nu$-Rules}
\label{further-opportunities}

We were motivated to develop a proof technique for extending
definitional equality with $\nu$-rules because there are many
opportunities where we might profit by doing so. Let us set out a
prospectus.

\paragraph{Reflexive coercion for type-based equality.} Altenkirch,
McBride and W. Swierstra developed a propositional equality for
intensional type theory~\cite{ObsEq} which differs from the usual inductive
definition ($\mathtt{refl~a~:~a~=~a}$) in that its main eliminator
\[
\inferrule{S,T : \mathtt{Set}\quad Q : S = T\quad s : S}
          {s[Q:S=T\rangle : T}
\]
computes by structural recursion first on the \emph{types}
$S$ and $T$, and then (where appropriate) on $s$, rather than
by pattern matching on the proof $Q$. Equality is still
reflexive, so evaluation can leave us with terms
$n[\mathtt{refl}\:n:N=N\rangle : N$
where $n$ is a neutral term in a neutral type $N$. It is clearly
a nuisance that this term does not compute to $n$, as would happen
if the eliminator matched on the proof. The fix is to add a
$\nu$-rule which discards coercions whenever it is type-safe to do so:
\[
  \fbb{\fba{$s$}$[Q:S=T\rangle$} = \fba{$s$}\qquad \mbox{if}\;S\equiv T:\mathtt{Set}
\]
It is easy to check that adding this rule for neutral terms makes it
admissible for all terms, and hence that we need add it not to
evaluation, but only to the reification process which follows, just
as with the $\nu$-rules in this paper. There, as here, this spares
the evaluation process from decisions which involve $\eta$-expansion
and thus require a name supply. The $\nu$-rule thus gives us a non-disruptive
means to respect the full computational behaviour of inductive equality
in the observational setting.

\paragraph{Functor laws.} Barral and Soloviev give a treatment of
functor laws for parametrized inductive datatypes by modifying the
$\iota$-rules of their underlying type
theory~\cite{DBLP:conf/csr/BarralS06}.  We should very much hope to
achieve the same result, as we did here in the special case of lists,
just by adding $\nu$-rules. Our preliminary
experiments~\cite{PigWeekNu} suggest that we can implement functor
laws once and for all in a type theory whose datatypes are given once
and for all by a syntactic encoding of strictly positive functors, as
Dagand and colleagues propose~\cite{LevitationPaper,ElabData}.  Moreover, Luo
and Adams have shown~\cite{DBLP:journals/mscs/LuoA08} that structural
subtyping for inductive types can be reified by a coherent system of
implicit coercions if functor laws hold definitionally.

\newcommand{\bind}{>\!\!>\!\!=}
\paragraph{Monad laws.} Watkins et al. give a definitional treatment
of monad laws in order to achieve
an adequate representation of concurrent processes encapsulated
monadically in a logical
framework~\cite{DBLP:conf/types/WatkinsCPW03}.
For straightforward free monads, an experimental extension of Epigram (by
Norell, as it happens)~\cite{PigWeekNu} suggests that we may readily allow
$\nu$-rules:
\[
\fbb{$\fba{$t$}\bind\mathtt{return}$} = \fba{$t$}
\qquad
\fbc{\(\fbb{\((\fba{$t$}\bind \sigma)\)}\bind\rho\)} =
\fbb{\(\fba{$t$}\bind((\bind \sigma)\cdot\rho)\)}
\]
Atkey's Foveran system uses a
similar normalization method for free monad laws~\cite{FusionMonad},
again for an encoded universe of underlying functors.

\paragraph{Decomposing functors.} Dagand and colleagues further
note that their syntax of descriptions for indexed functors is, by virtue of being
a syntax, itself presentable as the free monad of a functor. The
description decoder
\[
  \mathtt{Decode} : \mathtt{IDesc}\:I \to (I \to \mathtt{Set})
  \to \mathtt{Set}
\]
is structurally recursive in the description and lifts pointwise to
an interpretation of substitutions in the \(\mathtt{IDesc}\) monad
\[\begin{array}{l}
  \llbracket\_\rrbracket :
  (O \to \mathtt{IDesc}\:I) \;\;\to\;\;
     (I \to \mathtt{Set}) \to  (O \to \mathtt{Set}) \\
  \llbracket \sigma \rrbracket\:X\:o = \mathtt{Decode}\:(\sigma\:o)\:X
\end{array}\]
as indexed
functors with a `map' operation satisfying functor laws. However, not
only does this interpretation \emph{deliver} functors, it is
\emph{itself} a contravariant functor: the identity substitution yields
the identity functor just by $\beta\delta\iota$, but we may also
interpret Kleisli composition as reverse functor composition
\[
  \llbracket (>\!\!>\!\!=\sigma)\cdot\rho \rrbracket =
    \llbracket \rho \rrbracket \cdot \llbracket \sigma \rrbracket
\]
by means of a $\nu$-rule
\[
  \fbc{Decode~\fbb{$(\fba{$D$}>\!\!>\!\!=\sigma)$}~$X$} =
  \fbb{Decode~\fba{$D$}~$(\llbracket\sigma\rrbracket\:X)$}
\]
taking each $D$ to be some $\rho\:o$. If we want to do a `scrap your
boilerplate' style traversal of some described container-like
structure, we need merely exhibit the decomposition of the description as some
$(>\!\!>\!\!=\sigma)\cdot\rho$, where $\rho$ describes the invariant
superstructures and $\sigma$ the modified substructures, then invoke
the functoriality of $\llbracket \rho \rrbracket$. This $\nu$-rule
thus lets us expose functoriality over substructures not anticipated by explicit
parametrization in datatype declarations. We thus recover the kind of
ad-hoc data traversal popularized by L\"ammel and Peyton
Jones~\cite{DBLP:conf/tldi/LammelJ03} by static structural means.

\paragraph{Universe embeddings.} A type theory with
inductive-recursive definitions is powerful enough to encode universes
of dependent types by giving a datatype of codes \emph{in tandem} with
their interpretations~\cite{DBLP:conf/tlca/DybjerS99}, the paradigmatic
example being
\[\begin{array}{@{}l@{\;}|@{\;}l@{}}
\mathtt{U}_1 : \mathtt{Set}
  & \mathtt{El}_1 : \mathtt{U}_1 \to \mathtt{Set} \\
\mathtt{`Pi}_1 : (S:\mathtt{U}_1)\to
  & \mathtt{El}_1\:(\mathtt{`Pi}_1\:S\:T) = \\
\qquad\qquad(\mathtt{El}_1\:S\to \mathtt{U}_1 )\to\mathtt{U}_1 &
\;\;(s : \mathtt{El}_1\:S)\to
  \mathtt{El}_1\:(T\:s) \\
\vdots & \vdots
\end{array}\]
Palmgren~\cite{Palmgren:universes} suggests that one way to model a
cumulative hierarchy of such universes is to give each a code in the
next, so
\[\begin{array}{@{}l@{\;}|@{\;}l@{}}
\mathtt{U}_2 : \mathtt{Set}
  & \mathtt{El}_2 : \mathtt{U}_2 \to \mathtt{Set} \\
\mathtt{`U}_1 : \mathtt{U}_2
  & \mathtt{El}_2\:\mathtt{`U}_1 = \mathtt{U}_1 \\
\mathtt{`Pi}_2 : (S:\mathtt{U}_2)\to
  & \mathtt{El}_2\:(\mathtt{`Pi}_2\:S\:T) = \\
\qquad\qquad(\mathtt{El}_2\:S\to \mathtt{U}_2 )\to\mathtt{U}_2 &
\;\;(s : \mathtt{El}_2\:S)\to
  \mathtt{El}_2\:(T\:s) \\
\vdots & \vdots
\end{array}\]
and then define an embedding recursively
\[\begin{array}{l}
\uparrow : \mathtt{U}_1 \to \mathtt{U}_2 \\
\uparrow(\mathtt{`Pi}_1\:S\:T) =
  \mathtt{`Pi}_2\:(\uparrow S)\:(\lambda s.\:\uparrow(T\:s))
\end{array}\]
but a small frustration with this proposal is that $s$ is abstracted at
type $\mathtt{El}_2\:(\uparrow S))$, but used at type
$\mathtt{El}_1\:S$, and these two types are not definitionally equal
for an abstract $S$. One workaround is to make $\uparrow$ a
constructor of $\mathtt{U}_2$, at the cost of some redundancy of
representation, but now we might also consider fixing the discrepancy
with a $\nu$-rule
\[
  \fbc{$\mathtt{El}_2\:\fbb{$(\uparrow \fba{$S$})$}) $} =
    \fbb{$\mathtt{El}_1\:\fba{$S$}$}
\]
This is peculiar for our examples thus far, in that the $\nu$-rule is
needed even to typecheck the $\delta\iota$-rules for $\uparrow$,
reflecting the fact that $\uparrow$ should not be any old function
from $\mathtt{U}_1$ to $\mathtt{U}_2$, but rather one which preserves
the meanings given by $\mathtt{El}_1$ and $\mathtt{El}_2$. In effect,
the $\nu$-rule is expressing the coherence property of a richer notion
of morphism. It is inviting to wonder what other notions of coherence
we might enable and enforce by checking that $\nu$-rules hold of the
operations we implement.

\paragraph{Non-examples.} A key characteristic of a $\nu$-rule is that
it is a nut-preserving rearrangement of neutral term layers. Whilst
this is good for associativity and sometimes for distributivity, it is
perfectly useless for commutativity. Suppose $+$ for natural numbers
is recursive on its first argument, and observe that rewriting $x + y$
to $y + x$ when $x$ is neutral will not result in a neutral term
unless $y$ is also neutral. Less ambitious rules such as
$x+\mathtt{suc}\:y = \mathtt{suc}\:(x+y)$ and $x*0 = 0$ similarly make
neutral terms come unstuck, and so cannot be postponed until
reification if we want to be sure that evaluation suffices to show
whether any expression in a datatype can be put into
constructor-headed form. Walukiewicz-Chrzaszcz has proposed a more
invasive adoption of rewriting for Coq, necessitating a modified
evaluator, but incorporating rules which can expose
constructors~\cite{DBLP:journals/jfp/Walukiewicz-Chrzaszcz03}. Her
untyped rewriting approach sits awkwardly with $\eta$-laws, but we can
find a more carefully structured compromise.

\section{Discussion}

We fully expect to scale this technology up to type theory. Abel and
Dybjer (with Aehlig~\cite{NbeDep1} and T. Coquand~\cite{NbeDep2})
have already given normalization by evaluation algorithms which we
plan to adapt.

Finding good criteria for checking that candidate $\nu$-rules can safely
be added is of the utmost importance. We want to let the
programmer negotiate the new $\nu$-rules she wants, as long as the
machine can check that they yield a notion of standard form and lift
from neutral terms to all terms by the prior equational theory.

It is also interesting to try to integrate $\nu$-rules with more
practical presentations of normalization. For instance Gr\'{e}goire
and Leroy's conversion by compilation to a bytecode machine derived
from Ocaml's ZAM~\cite{ConvTestZAMed} decides $\eta$ by expansion only
when provoked by a $\lambda$: such laziness is desirable when possible
but causes trouble with $\eta$-rules for unit types and may conceal
the potential to apply $\nu$-rules. Hereditary
substitution~\cite{DBLP:conf/types/WatkinsCPW03}, formalized by
Abel~\cite{HeredSubst1} and by Keller and Altenkirch~\cite{HeredSubst2},
may be easier to adapt.

\section*{Acknowledgements}

We would like to thank the anonymous reviewers for their helpful comments and
suggestions as well as Stevan Andjelkovic for carefully reading our draft.

%~ \appendix
%~ \section{Appendix Title}
%~
%~ This is the text of the appendix, if you need one.
%~
%~ \acks
%~
%~ Acknowledgments, if needed.

% We recommend abbrvnat bibliography style.

\bibliographystyle{abbrvnat}
\bibliography{main}

% The bibliography should be embedded for final submission.

%~ \begin{thebibliography}{}
%~ \softraggedright
%~ \bibitem[Smith et~al.(2009)Smith, Jones]{smith02}
%~ P. Q. Smith, and X. Y. Jones. ...reference text...
%~
%~ \end{thebibliography}

\end{document}